\documentclass[12pt]{article}
\usepackage{stmaryrd}
\usepackage[dvips]{graphicx}
\usepackage{amssymb,amsmath,amsfonts,palatino,amsthm}
\usepackage{amssymb,mathtools}
\usepackage[sans]{dsfont}
\newcommand{\beq}{\begin{equation}}
\newcommand{\eeq}{\end{equation}}
\newcommand{\beqq}{\begin{equation*}}
\newcommand{\eeqq}{\end{equation*}}

\newcommand{\f}{\begin{equation}}
\newcommand{\ff}{\end{equation}}

\newtheorem{theorem}{Theorem}

\newtheorem{lemma}{Lemma}
\newtheorem{conjecture}{Conjecture}
\begin{document}

\title{Conserved Quantities and the Algebra of Braid Excitations in Quantum Gravity}
\author{Song He\thanks{Email address:
hesong@pku.edu.cn}\\ \\School of Physics, Peking University,
Beijing, 100871, China, \\ \\ \\Yidun Wan\thanks{Email address:
ywan@perimeterinstitute.ca}
\\ \\
Perimeter Institute for Theoretical Physics,\\
31 Caroline st. N., Waterloo, Ontario N2L 2Y5, Canada, and \\
Department of Physics, University of Waterloo,\\
Waterloo, Ontario N2J 2W9, Canada\\}
\date{February 15, 2008}
\maketitle
\vfill
\begin{abstract}
We derive conservation laws from interactions of braid-like
excitations of embedded framed spin networks in Quantum Gravity. We
also demonstrate that the set of stable braid-like excitations form
a noncommutative algebra under braid interaction, in which the set
of actively-interacting braids is a subalgebra.
\end{abstract}
\vfill
\newpage
\tableofcontents
\newpage

\section{Introduction}
Recently, there has been a significant amount of work done towards a
quantum theory of gravity with matter as topological
invariants\cite{Bilson-Thompson2006, Hackett2007, LouNumber,
Wan2007, LeeWan2007, HackettWan2008, LeeWan2008, HeWan2008a}.
\cite{Bilson-Thompson2006, Hackett2007, LouNumber} work on framed
three-valent spin networks present in models related to Loop Quantum
Gravity with non-zero cosmological constant\cite{Major1995,
Smolin2002}, in which the topological invariants of ribbon braids
are able to detect chirality and code chiral conservation laws.
However, the results of this approach have a serious limitation in
the sense that there is no dynamics of the conserved
quantities\cite{Hackett2007}.

To resolve this limitation, a new approach based on embedded
four-valent spin networks is proposed in \cite{Wan2007}, and is
shown to have dynamics built in by means of the so-called dual
Pachner moves\cite{LeeWan2007}. Here the four-valent spin networks
can be understood as those naturally occur in spin foam
models\cite{spin-foam}, or in a more generic context as the original
proposal of spin networks put forward by Penrose\cite{Penrose}, plus
embedding.

The dynamical objects found by the new approach are three-strand
braids, each of which is formed by three common edges of two
adjacent nodes of the embedded four-valent spin network. The stable
three-strand braids, under certain stability condition, are local
excitations\cite{LeeWan2007, Isabeau2008}. Among all stable braids,
there is a small class of braids which are able to propagate on the
spin network. The propagation of these braids are chiral, in the
sense that some braids can only propagate to their left with respect
to the local subgraph containing the braids, while some only
propagate to their right and some do both\cite{Wan2007, LeeWan2007}.
There is another small class of braids, the actively-interacting
braids (hereafter called "active braids" for short); each is two-way
propagating and is able to merge with its neighboring braid when the
interaction condition is met\cite{LeeWan2007}. In the sequel, braids
that are not active are called passive, including stationary braids,
i.e. those do not even propagate.

\cite{Wan2007, LeeWan2007} are based on a graphic calculus developed
therein. However, although the graphic calculus has its own
advantages - in particular in describing, e.g. the full procedure of
the propagation a braid, it is not very convenient for finding
conserved quantities of a braid which are useful to characterize the
braid as a matter-like local excitation. In view of this,
\cite{HackettWan2008} proposed an algebraic notation of the active
braids and derived conserved quantities by means of the new
notation.

To these ends, in this paper, we generalize the algebraic notation
in\cite{HackettWan2008} to the case of generic three-strand braids.
Within this notation, the algebraic equivalence moves are defined
and the quantities conserved under these are identified. Finally the
algebra of interactions between active braids and passive braids is
discussed. This leads to the following results:
\begin{enumerate}
\item There exist conserved quantities under interactions and we are
able to show the form of these conservation laws.
\item Precise algebraic forms of braid interactions are presented.
\item The set of all stable braids form an algebra under braid
interaction, in which the set of all active braids is the
subalgebra.
\item This algebra is noncommutative due to the fact that the left
and right interactions of an active braid onto another braid are not
the same in general. Conditions of commutative interactions are
explicitly given.
\item Asymmetric interactions can be related by discrete
transformation, such as P, T, CP, and CT.
\end{enumerate}

An immediate application of these results is realized in a companion
paper\cite{HeWan2008a} which discovers the C, P, and T
transformations of braids by means of conserved quantities found in
\cite{HackettWan2008} and in this paper.

\section{Notation}%
We will extend the algebraic notation of active braids to the
general case, namely to propagating braids and in fact all braids.
However, for illustrative purposes we keep the graphic notation
wherever necessary. We adopt the graphical notation we proposed in
\cite{Wan2007, LeeWan2007}. A generic 3-strand braid is shown in
Fig. \ref{braid}(a), while a concrete example is depicted in Fig.
\ref{braid}(b). More precisely, what are shown in Fig. \ref{braid}
are braid diagrams as projections of the true 3-strand braids
embedded in a topological three manifold. Each spin network can be
embedded in various ways, some of which are diffeomorphic to each
other. The projection of a specific embedding of a braid is called a
braid diagram; many braid diagrams are equivalent and belong to the
same equivalence class, in the sense that they correspond to the
same braid and can be transformed into each other by equivalence
moves\cite{Wan2007}. Thus a braid refers to the whole equivalence
class of its braid diagrams. However, one can choose a braid diagram
of an equivalence class as the representative of the class, we
therefore will not distinguish a braid from a braid diagram in the
sequel unless an ambiguity arises. Besides, a braid always means a
3-strand braid.

\begin{figure}[h]
\begin{center}
\includegraphics[
height=2.1318in,
width=2.8349in ]{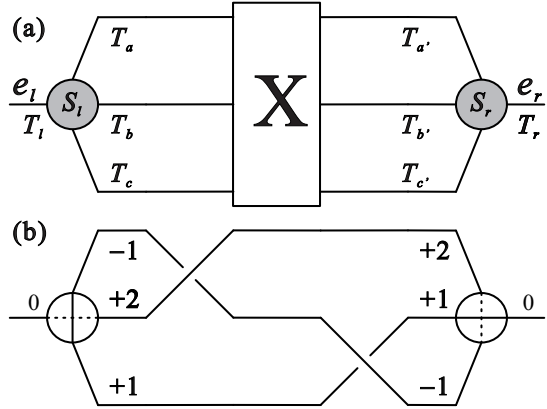}%
\caption{(a) is a generic 3-strand braid diagram formed by the three
common edges of two end-nodes. $S_l$ and $S_r$ are the states of the
left and right \textbf{end-node}s respectively, taking values in $+$
or $-$. $X$ represents a sequence of crossings, from left to right,
formed by the three strands between the two nodes. $T_a$, $T_b$, and
$T_c$ are the \textbf{internal twists} respectively on the three
strands from top to bottom, on the left of $X$. $T_l$ and $T_r$,
called \textbf{external twists}, are respectively on the two
external edges $e_l$ and $e_r$. All twists are valued in
$\mathbb{Z}$ in units of $\pi/3$\cite{Wan2007}. (b) is a concrete
example of a braid diagram, in which the left end-node is in the
'$+$' state while the right
end-node is in the '$-$' state.}%
\label{braid}
\end{center}
\end{figure}

It is important to emphasize the choice of the representative of an
equivalence class of braid diagrams. In \cite{LeeWan2007}, each
equivalence class of braid diagrams is represented by its unique
element which has zero external twists (see Fig. \ref{braid}(b) for
an example). This choice makes the propagation and interaction of
braids defined in\cite{LeeWan2007} easier to handle. However, there
are three types of stable braids, viz active braids, propagating
braids, and stationary braids\cite{LeeWan2007, LeeWan2008}.
\textbf{Propagating} braids are able to exchange places with their
adjacent substructures in the graph under the local dynamical moves,
whereas the \textbf{stationary braids} cannot propagate in the way
the propagating braids do. These braids are in most cases
represented by braid diagrams of zero external twists.

On the other hand, as pointed out in \cite{LeeWan2007,
HackettWan2008}, the \textbf{active} braids, each of which can
propagate and can interact onto any other braid in the sense that it
can merge with another adjacent braid to form a new braid as long as
the interaction condition is met, are happen to be both
\textbf{completely left-} and \textbf{right-reducible}, i.e. such a
braid is always equivalent to a trivial braid diagram with possibly
twists on its three strands and two external edges. Thus it is more
convenient to represent each of these braids by a trivial braid
diagram (which is not unique) of the corresponding equivalence
class. In fact, \cite{HackettWan2008} chose this representation and
derived conserved quantities of this type of braids under
interaction by introducing an algebraic notation of them and a
symbolic way of handling the interactions of them. This trivial
representation of active braids is actually a special case of the
so-called extremal representation of generic braids\cite{Wan2007}. A
braid diagram as an extremal representation of a braid is called an
\textbf{extremum} of the braid. The name-extremum-manifests per se
its meaning: a braid diagram with least number of crossings, among
all braid diagrams in the same equivalence class.

Therefore, our generalized algebraic notation of a braid should take
care of all possible choices of representative of a braid,
especially the aforementioned special representations. Note that,
the class of propagating braids excludes any active braid, which is
also propagating though.

\begin{figure}
[h]
\begin{center}
\includegraphics[
height=0.9574in, width=2.8193in
]%
{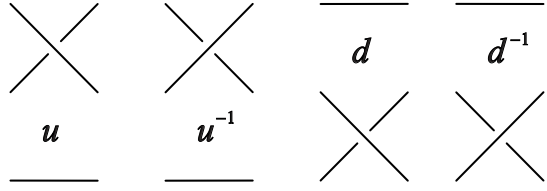}%
\caption{The generators of braid group $B_3$, and hence of the
crossing sequence of a generic 3-strand braid.}%
\label{Xing}
\end{center}
\end{figure}

Let us now concentrate on the generic case depicted in Fig.
\ref{braid}(a). Apart from the internal twists, the interior of a
braid, which is the region between the two end-nodes and is
characterized by the sequence of crossings, satisfies the definition
of an ordinary braid, arranged horizontally. We can thus denote the
sequence of crossings $X$ by generators of the braid group $B_3$.
The group $B_3$ has two generators and their inverses. Since we
arrange a braid diagram horizontally, the generator and its inverse
formed by the upper two strand of a braid are named $u$ and $u^{-1}$
respectively, while the one and its inverse formed by the lower two
strand of a braid are $d$ and $d^{-1}$ respectively. This convention
is illustrated in Fig. \ref{Xing}. Then for example, the crossing
sequence in Fig. \ref{braid}(b) reads $X=u^{-1}d$, from left to
right. We also assign an integral value, the crossing number, to
each generator, i.e. $u=d=1$ and $u^{-1}=d^{-1}=-1$.

For an arbitrary sequence $X$ of \textbf{order} $n=|X|$, namely the
number of crossings, we can write $X=x_1x_2\cdots x_i\cdots x_n$,
where $x_i\in\{u,u^{-1},d,d^{-1}\}$ represents the $i$-th crossing
from the left. Therefore, each $x_i$ in $X$ has a two-fold meaning:
on the one hand, it is an abstract crossing; on the other hand, it
represents an integral value, $1$ or $-1$. When a $x_i$ appears in a
multiplication it is usually understood as an abstract crossing,
while in a summation it is normally an integer. Note that, as
generators of the group $B_3$, the generators of $X$ obey the
following equivalence relations.
\begin{equation}
\begin{aligned}
udu^{-1}&=d^{-1}ud\\
u^{-1}du&=dud^{-1}\\
udu&=dud
\end{aligned}
\end{equation}
We assume in any $X$, the above equivalence relations have been
applied to remove any pair of a crossing and its inverse. For
example, the sequence $udu^{-1}d^{-1}$ should have been written as
$udu^{-1}d^{-1}=d^{-1}udd^{-1}=d^{-1}u$ by the first relation above.
The crossing sequence $X$ clearly induces a permutation, denoted by
$\sigma_X$, of the three strands of a braid. It is obvious that the
induced permutation $\sigma_X$ takes value in $S_3$, the permutation
group of the set of three elements. In terms of disjoint cycles,
\beq S_3=\{\mathds{1},(1\ 2), (1\ 3), (2\ 3), (1\ 2\ 3), (1\ 3\
2)\}.\label{eqS3}\eeq More precisely, the three internal twists in a
triple $(T_a,T_b,T_c)$ on the left of the $X$ are permutated by the
induced permutation into
$(T_a,T_b,T_c)\sigma_X=(T_{a'},T_{b'},T_{c'})$, where
$(T_{a'},T_{b'},T_{c'})$ is the triple of the internal twists on the
right of $X$. Here $\sigma_X$ is defined to be a \emph{left-acting}
function on the triple of internal twists for two reasons. Firstly,
this is a convention of permutation group. Secondly, when there is
another crossing sequence, say $X'$, appended to the right of $X$,
which usually happens in interactions of braids, we naturally have
$(T_a,T_b,T_c)\sigma_X\sigma_{X'}=(T_a,T_b,T_c)\sigma_{XX'}$, such
that the induced permutation of the newly appended crossings is
applied after the action of $\sigma_X$.

On the other hand,
$(T_a,T_b,T_c)=\sigma^{-1}_X(T_{a'},T_{b'},T_{c'})$, where
$\sigma^{-1}_X$ is the inverse of $\sigma_X$ and is a
\emph{right-acting} function of the triple. One should keep in mind
that the indices of internal twists such as $T_a$ and $T_a'$, $a$
and $a'$ are abstract and have no meaning until their values and
positions in the triple of internal twists are fixed. So $(T_a, T_b,
T_c)=(T_d, T_e, T_f)$ means respectively $T_a=T_d,T_b=T_e$, and
$T_c=T_f$. In the rest of the paper, we will also consider the
addition of two triples of twists, i.e. $(T_a, T_b, T_c)+(T_d, T_e,
T_f)=(T_a+T_d, T_b+T_e, T_c+T_f)$. Therefore, we can denote a
generic braid as Fig. \ref{braid}(a) by
$$\left._{T_l}^{S_l}\hspace{-0.5mm}[(T_a,T_b,T_c)\sigma_X]^{S_r}_{T_r}\right.,$$
or equally well by
$$\left._{T_l}^{S_l}\hspace{-0.5mm}[\sigma^{-1}_X(T_{a'},T_{b'},T_{c'})]^{S_r}_{T_r}\right..$$
In such a way, which side of the crossing sequence a triple of
internal twists is on is transparent. For instance, the braid in
Fig. \ref{braid}(b) can be written as
$\left._{\hspace{0.5mm}0}^+\hspace{-0.2mm}[(-1,2,1)\sigma_{u^{-1}d}]^-_0\right.$
or
$\left._{\hspace{0.5mm}0}^+\hspace{-0.2mm}[\sigma^{-1}_{u^{-1}d}(2,1,-1)]^-_0\right.$,
where $\sigma_{u^{-1}d}=(1\ 3\ 2)$ and $\sigma^{-1}_{u^{-1}d}=(1\ 2\
3)$.

It is now manifest that a generic braid is characterized by the
8-tuple, $\{T_l,S_l,T_a,T_b,T_c,\\X,S_r,T_r\}$. As mentioned before,
$S_l$ and $S_r$ are just signs, $+$ or $-$, such that $-(+)=-$ and
$-(-)=+$. Hence, for an arbitrary end-node state $S$, we may use
both $-S$ and $\bar{S}$ for the inverse of $S$. In fact, this
8-tuple is not completely arbitrary for different type of braids.
For a propagating braid $B$ of $n$ crossings represented by the
braid diagrams with no external twists we have the following
constraints.
\begin{enumerate}
\item $T_l=T_r=0$.
\item The triple $(S_l,X,S_r)$ is not arbitrary. If $B$ is (\textbf{left-})
\textbf{right-propagating}, according to \cite{Wan2007, LeeWan2007},
$B$ must be (\textbf{left-}) \textbf{right-reducible}, in particular
its first crossing on the (left) right can be eliminated by the
equivalence move, a $\pi/3$ rotation which also flips the (left)
right end-node. That is, letting $X=x_1x_2\cdots x_{n-1}x_n$, then
under a $\pi/3$ rotation on the (left) right end-node, the triple
$(S_l,X,S_r)$ becomes ($(\bar{S_l}, x_2\cdots x_{n-1}x_n, S_r)$)
$(S_l, x_1x_2\cdots x_{n-1}, \bar{S_r})$.
\item The triple $(T_a,T_b,T_c)$ is not arbitrary; however, the
general pattern of them, ensuring the propagation of $B$, has not
yet been found and is under investigation. But the algebra
formulated and the conserved quantities to be found in this paper
may turn out to be helpful to resolve this problem.
\end{enumerate}
Any braid whose characterizing 8-tuple violates the above
constraints is not propagating

It is useful in certain situation to represent a propagating braid
by its \textbf{extrema} in the corresponding equivalent class, which
are the braid diagrams of the least number of crossings, obtained by
rotations of the unique representative with zero external twists,
getting rid of the reducible crossings\cite{Wan2007}. However, we
would like to leave the discussion of this representation to the
next section after we defined rotations symbolically.

For an active braid, we choose to represent it by its extrema, i.e.
its trivial braid diagrams with external twists. This has actually
been carried out in \cite{HackettWan2008}, we thus will not repeat
any detail here. It is good to see our notation of the generic case
reduces to that of the active braids defined in
\cite{HackettWan2008}. For active braids represented by trivial
diagrams, the crossing sequence is trivial and hence the induced
permutation is the identity, i.e. $\sigma_X=\mathds{1}$. There is no
difference between the triple of internal twists on the left of $X$
and the one on the right. Moreover, we have $S_l=S_r=S$ in this
case\cite{HackettWan2008}. As a result, the generic notation
uniquely boils down to
$$\left._{T_l}^{S}\hspace{-0.5mm}[T_a,T_b,T_c]^{S}_{T_r}\right.,$$
which is the very notation introduced in \cite{HackettWan2008} for
active braids in their trivial representations.

\section{Algebra of equivalence moves: symmetries and relations}%
\label{secEquiv}%
Since an active braid can always be reduced to trivial braids with
twists, it is sufficient to discuss the algebra of simultaneous
rotations for these trivial braids. In \cite{HackettWan2008}, we
have found general effects of simultaneous rotations on trivial
braids, especially conserved quantities under this class of
equivalence moves. For generic braids, however, we need to consider
more generalized rotations of a braid, denoted by $R_{m,n}$ with
$m,n\in\mathbb{Z}$, which is the combination of an $m\pi/3$ rotation
on the left end-node and an $n\pi/3$ one of the right end-node of
the braid. Let us record the algebraic form of such a rotation on a
generic braid, and then explain it,
\begin{align}
&R_{m,n}(_{T_l}^{S_l}\hspace{-0.5 mm}[ (T_a,T_b,T_c)\sigma_X] _{T_r}^{S_r})\nonumber\\
&=_{\ \ T_l+ m}^{(-)^mS_l}\hspace{-1.5 mm}[
(P^{S_l}_m(T_a-m-n,T_b-m-n,T_c-m-n))\sigma_{X_l(S_l,m)XX_r(S_r,n)}]
_{T_r+ n}^{(-)^nS_r}.%
\label{eqRotAlgebraic1}
\end{align}

On the RHS of Eq. \ref{eqRotAlgebraic1}, the original end-node
states, $S_l$ and $S_r$, become $(-)^m S_l$ and $(-)^n S_r$
respectively. This is because by \cite{Wan2007}, a $\pi/3$ rotation
of a node always flips the state of the node once, which means a
$m\pi/3$ rotation should flip the state of a node $m$ times. Also
according to \cite{Wan2007}, a rotation of the left (right) end-node
of a braid creates a crossing sequence, appended to the left (right)
of the original crossing sequence of the braid. In Eq.
\ref{eqRotAlgebraic1}, the newly-generated sequence on the left is
denoted by a function $X_l(S_l,m)$, depending on the original left
end-node state and the amount of rotation, $m$. Likewise, the new
crossing sequence on the right is denoted by the function
$X_r(S_r,n)$. We will elaborate these two functions shortly. As a
consequence, the induced permutation by the crossing sequence
changes accordingly, from $\sigma_X$ to
$\sigma_{X_l(S_l,m)XX_r(S_r,n)}$.

In addition, the left triple of internal twists is affected by the
rotation of the left end-node, which induces a permutation
$P^{S_l}_m$ on the triple, determined by the original end-node state
and the amount of rotation $m$. This function, which obviously takes
its value in the group $S_3$ shown in Eq. \ref{eqS3}, is the same as
the that induced by a simultaneous rotation on active braids,
defined in \cite{HackettWan2008}. One may wonder why the similar
permutation induced by the rotation on the right end-node does not
appear in Eq. \ref{eqRotAlgebraic1}. This is due to the advantage of
our notation which needs the triple of internal twists on one side,
while the triple on the other side is taken care of by the
permutation $\sigma$. If one indeed wants to have the triple of
internal twists beside the right end-node explicit, one can use the
alternative form of the rotation as follows instead.
\begin{align}
&R_{m,n}(_{T_l}^{S_l}\hspace{-0.5 mm}[\sigma^{-1}_X (T_{a'},T_{b'},T_{c'})] _{T_r}^{S_r})\nonumber\\
&=_{\ \ T_l+ m}^{(-)^mS_l}\hspace{-1.5 mm}[
\sigma^{-1}_{X_l(S_l,m)XX_r(S_r,n)}(P^{S_r}_{-n}(T_{a'}-m-n,T_{b'}-m-n,T_{c'}-m-n))]
_{T_r+ n}^{(-)^nS_r},%
\label{eqRotAlgebraic2}
\end{align}

Finally, the common increment of $-m-n$ of all internal twists, and
the changes of the two external twists under the rotation $R_{m,n}$
in Eq. \ref{eqRotAlgebraic1} and Eq. \ref{eqRotAlgebraic2} are
simple effects of the rotation\cite{Wan2007}.

We now explain more about these functions. Since the permutations
$P^S_m$ here are the same as those defined in \cite{HackettWan2008},
we adopt the following lemma from \cite{HackettWan2008} which states
the general relations they satisfy; a proof of this lemma can be
found in the reference.
\begin{lemma}
\begin{align*}
P^S_{2n}P^S_{-2n} &\equiv\mathds{1}\\
P^S_{2n+1}P^{\bar{S}}_{-(2n+1)} &\equiv\mathds{1}\\
P^S_{n} &\equiv P^{\bar S}_{-n},
\end{align*}
where $n\in\mathbb{Z}$. \label{lemmPermRel}
\end{lemma}

Besides, the equations below are easy to
derive\cite{HackettWan2008}; they are listed here for possible
future use.
\begin{equation}
\begin{aligned}
P^+_1=P^-_{-1} &=(1\ 2)\\
P^+_{-1}=P^-_1 &=(2\ 3)\\
P^+_2=P^-_{-2} &=(1\ 3\ 2)\\
P^+_{-2}=P^-_2 &=(1\ 2\ 3)\\
P^{\pm}_{6n+3} &\equiv(1\ 3)\\
P^{\pm}_{6n} &\equiv\mathds{1},
\end{aligned}
\label{eqPermrelation}
\end{equation}
where $n\in\mathbb{Z}$.

According to the graphic definitions of rotations in \cite
{Wan2007}, we found that $X_l(S,m)$ and $X_r(S,n)$ have the
following general algebraic forms.
\begin{eqnarray}
X_l(+,m)&=
\begin{cases}  (ud)^{-m/2}& \text{if $m$ is even,}
\\
d(ud)^{(-1-m)/2} &\text{if $m$ is odd.}
\end{cases}
X_l(-,m)=
\begin{cases}  (du)^{-m/2}& \text{if $m$ is even,}
\\
u(du)^{(-1-m)/2} &\text{if $m$ is odd.}
\end{cases}\\
X_r(+,n)&=
\begin{cases}  (ud)^{-n/2}& \text{if $n$ is even,}
\\\nonumber
(ud)^{(1-n)/2}d^{-1} &\text{if $n$ is odd.}
\end{cases}
X_r(-,n)=
\begin{cases}  (du)^{-n/2}& \text{if $n$ is even,}
\\
(du)^{(1-n)/2}u^{-1} &\text{if $n$ is odd.}
\end{cases}
\label{eqCrossing}
\end{eqnarray}
where $n,m\in\mathbb{Z}$. If an exponent in Eq. \ref{eqCrossing} is
positive, it means, for example, $(ud)^2=udud$. We utilize a
definition in \cite{HeWan2008a}, which is, for a crossing sequence
$X=x_1\cdots x_i\cdots x_N,\ N\in\mathbb{N}$, $X^{-1}=x^{-1}_N\cdots
x^{-1}_i\cdots x^{-1}_1$. Given this, the meaning of the negative
exponents in Eq. \ref{eqCrossing} is clear. For instance,
$(ud)^{-2}=d^{-1}u^{-1}d^{-1}u^{-1}$.

It is obvious that the number of crossings of either $X_l(S_l,m)$ or
$X_r(S_r,m)$ does not depend on the end-node state,
\begin{equation}
|X_l(+,m)|=|X_l(-,m)|=|X_r(+,m)|=|X_r(-,m)|=|m|, \label{eqRotNcro}
\end{equation}
neither does the sum of crossing numbers of $X_l(S_l,m)$ or
$X_r(S_r,m)$, namely
\begin{equation}
\sum\limits_{i=1}^{|m|}x_i
=\sum\limits_{i=1}^{|m|}y_i=\sum\limits_{i=1}^{|m|}z_i=\sum\limits_{i=1}^{|m|}w_i=-m.
\label{eqRotXvalue}
\end{equation}
where, $x_i\in X_l(+,m)$, $y_i\in X_l(-,m)$, $z_i\in X_r(+,m)$, and
$w_i\in X_r(-,m)$. In addition, there is a useful relation between
$X_l(S,m)$ and $X_r(S,n)$, as stated in the following Lemma.
\begin{lemma}
$X_l(S,m)X_r(S,-m)\equiv\mathbb{I}$. $\mathbb{I}$ stands for no
crossing.
\label{lemmCrossingRelation}
\end{lemma}

\begin{proof}
For $m$ even,
\begin{eqnarray*}
X_l(+,m)X_r(+,-m)=(ud)^{-m/2}(ud)^{m/2}=\mathbb{I},\\
X_l(-,m)X_r(-,-m)=(du)^{-m/2}(du)^{m/2}=\mathbb{I};
\end{eqnarray*}
for $m$ odd,
\begin{eqnarray*}
X_l(+,m)X_r(+,-m)=d(ud)^{(-1-m)/2}(ud)^{(1+m)/2}d^{-1}=\mathbb{I},\\
X_l(-,m)X_r(-,-m)=u(du)^{(-1-m)/2}(du)^{(1+m)/2}u^{-1}=\mathbb{I}.
\end{eqnarray*}
In conclusion, for any $S$ and $m$, we have
\begin{equation*}
X_l(S,m)X_r(S,-m)\equiv\mathbb{I}.
\end{equation*}
\end{proof}

The rotation $R_{m,n}$ is actually a generalization of the
simultaneous rotation $R_{n,-n}$, defined in \cite{HackettWan2008}
as acting on actively-interacting braids in their extremal
representations, each of which is a trivial braid diagram with two
identical end-node states. In this case $S_l=S_r=S$, and
$X\equiv\mathbb{I}$, indicating $\sigma_X=\mathds{1}$. Therefore,
for consistency, our general rotation $R_{m,n}$ should reduce to the
simultaneous rotation $R_{m,-m}$ on these braids, if we set $n=-m$.
This is indeed so because
\begin{equation*}
R_{m,-m}(_{T_l}^{S}\hspace{-0.5 mm}[T_a,T_b,T_c]
_{T_r}^{S})=_{T_l+m}^{(-)^mS}\hspace{-0.5
mm}[(P^S_m(T_a,T_b,T_c))\Sigma_{X_l(S,m)\mathbb{I}X_r(S,m)}]
_{T_r-m}^{(-)^mS};
\end{equation*}
however, Lemma \ref{lemmCrossingRelation} gives
$X_l(S,m)\mathbb{I}X_r(S,m)\equiv\mathbb{I}$, such that
\begin{equation*}
R_{m,-m}(_{T_l}^{S}\hspace{-0.5 mm}[T_a,T_b,T_c]
_{T_r}^{S})=_{T_l+m}^{(-)^mS}\hspace{-0.5 mm}[P^S_m(T_a,T_b,T_c)]
_{T_r-m}^{(-)^mS},
\end{equation*}
which is the very simultaneous rotation defined in
\cite{HackettWan2008}.

With these ingredients, we can proceed to find out conserved
quantities of a generic braid under general equivalence moves,
$R_{m,n}$. One can see that unlike trivial braids under simultaneous
rotations in \cite{HackettWan2008}, $T_l+T_r$ and the triple
$(T_a,T_b,T_c)$ are no longer conserved here for a generic rotation,
only a combination of them with sum of crossing numbers, namely the
effective twist
$\Theta=T_l+T_r+\sum\limits_{i=a}^{c}T_i-2\sum\limits_{i=1}^{|X|}x_i$
is conserved under these general equivalence moves. Besides, the
conserved quantities  $S^2$ is generalized to the \textbf{effective
state} $\chi=(-)^{|X|}S_lS_r$ for generic braids. These results are
summarized as the following Lemma.

\begin{lemma}
Under a general rotation $R_{m,n}$, a braid's effective twist
number, $\Theta$, and its effective state, $\chi$, are conserved.
\label{lemmConserve}
\end{lemma}
\begin{proof}
By Eq. \ref{eqRotAlgebraic1}, a general rotation $R_{m,n}$ can
transform a generic braid
\begin{equation*}
_{T_l}^{S_l}\hspace{-0.5 mm}[(T_a,T_b,T_c)\sigma_X] _{T_r}^{S_r},
\end{equation*}
with
$\Theta=T_l+T_r+\sum\limits_{i=a}^{c}T_i-2\sum\limits_{i=1}^{|X|}x_i$
and $\chi=(-)^{|X|}S_lS_r$, into
\begin{equation*}
_{\ \ T_l+ m}^{(-)^mS_l}\hspace{-0.5 mm}[
(P^{S_l}_m(T_a-m-n,T_b-m-n,T_c-m-n))\sigma_{X_l(S_l,m)XX_r(S_r,n)}]
_{T_r+ n}^{(-)^nS_r},
\end{equation*}
with
\begin{equation*}
\Theta'=(T_l+m)+(T_r+n)+(\sum\limits_{i=a}^{c}T_i-3(m+n))-2(\sum\limits_{i=1}^{|m|}
y_i+\sum\limits_{i=1}^{|X|}x_i+\sum\limits_{i=1}^{|n|} z_i),
\end{equation*}
where $y_i\in X_l(S_l,m)$ and $z_i\in X_r(S_r,n)$ and
\begin{equation*}
\chi'=(-)^{|X_l(S_l,m)|+|X|+|X_r(S_r,n)|}(-)^mS_l(-)^nS_r.
\end{equation*}
Nonetheless, by Eq. \ref{eqRotNcro} and Eq. \ref{eqRotXvalue}, we
obtain
\begin{equation*}
\Theta'=T_l+T_r+\sum\limits_{i=a}^{c}T_i+\sum\limits_{i=1}^{|X|}x_i+m+n-3(m+n)-2(-m-n)=\Theta,
\end{equation*}
and
\begin{equation*}
\chi'=(-)^{|X|+|m|+|n|}(-)^mS_l(-)^nS_r=(-)^{|X|}S_lS_r=\chi.
\end{equation*}
This establishes the proof.
\end{proof}

As a direct consequence of Lemma \ref{lemmConserve}, we have the
following Theorem, which provides a character for
actively-interacting braids.
\begin{theorem}
The effective state of any actively-interacting braids is $\chi=1$,
and any braid with effective state $\chi=-1$ must be passive.
\label{theoChi}
\end{theorem}
\begin{proof}
The proof of this theorem is very simple. Any actively-interacting
braid has a trivial representation, whose effective state is
$\chi\equiv S^2=1$. Hence, according to Lemma \ref{lemmConserve},
the effective state of any actively-interacting braid must be
$\chi=1$. This implies, on the other hand, any stable braid with
$\chi=-1$ is never actively-interacting, and is thus passive.
\end{proof}

All results we have obtained so far are valid for braids in any
representation. Now we would like to consider the extremal
representation of a braid, and try to find out how equivalence moves
act on braids in this representation in particular, and the
corresponding conserved quantities. The case of actively-interacting
braids are investigated in \cite{HackettWan2008}, and it turns out
that there are infinite number of extrema which are trivial braids
related to each other by simultaneous rotations. Moreover,
$T_l+T_r$, the triple $(T_a,T_b,T_c)$ up to permutation, and $S^2$
are conserved under these simultaneous rotations. For passive
braids, the situation is more involved for that their extrema are
not trivial braids and that generic rotations (including generic
simultaneous rotations) increase the number of crossings of an
extremum. However, there are also infinite number of extrema of a
passive braid due to the following Lemma.

\begin{lemma}
Simultaneous rotations of the form $R_{3k,-3k}$ with
$k\in\mathbb{Z}$ takes an extremum of a braid to another extremum of the braid.%
\label{lemm3KsimRot}
\end{lemma}
\begin{proof}
By the definition of extremal, all extrema of the same braid have
the same number of crossings. Thus we only need to prove that
$R_{3k,-3k},\ k\in\mathbb{Z}$ on an extremum preserves its number of
crossings, then the resultant representation of braid must also be
an extremum; otherwise the braid diagram undergoing the rotation
should not be an extremum in the first place. Additionally, since
$R_{3k,-3k}=R_{\pm3,\mp3}^{|k|}$ by \cite{HackettWan2008}, where the
$\pm$ and $\mp$ depend on the sign of $k$. It is sufficient to prove
the case of $k=\pm1$, which are just simultaneous $\pi$ rotations.
We now prove that simultaneous $\pi$-rotations of a braid take the
braid's crossing sequence, $X=x_1\cdots x_i\cdots x_N,\
N\in\mathbb{N}$, to $\bar{X}=\bar{x}_1\cdots \bar{x}_i\cdots
\bar{x}_N,\
N\in\mathbb{N}$, where%
\begin{equation*}
\bar{x}_i=\left\{%
\begin{array}
[c]{l}%
d,\ x_i=u\\
u,\ x_i=d
\end{array}
\right..
\end{equation*}
and thus keep the number of crossings of the braid invariant. For
one-crossing braids with arbitrary $S_l$ and $S_r$, it is
straightforward to see that
$X_l(S_l,\pm3)x_1X_r(S_r,\mp3)=\bar{x}_1$, for
$x_1=u,d,u^{-1},d^{-1}$. We assume that the this is true for any
braid with up to $N\in\mathbb{N}$ crossings, with arbitrary end-node
states, namely
\begin{equation*}
X_l(S_l,\pm3)x_1x_2...x_NX_r(S_r,\mp3)=\bar{x}_1\bar{x}_2...\bar{x}_N.
\end{equation*}
Hence, for any braid with $N+1$ crossings and end-node states, say
$S'_l$ and $S'_r$,
\begin{align*}
&X_l(S'_l,\pm3)x_1x_2...x_Nx_{N+1}X_r(S'_r,\mp3)\\
\xRightarrow{X_r(S,\mp3)X_l(S,\pm3)=\mathbb{I}}
&=(X_l(S'_l,\pm3)x_1x_2...x_NX_r(S,\mp3))(X_l(S,\pm3)x_{N+1}X_r(S'_r,\mp3)),
\end{align*}
where we inserted a pair of crossing sequences, $X_r(S,\pm3)$ and
$X_l(S,\pm3)$ with arbitrary $S$, whose product is trivial by Lemma
\ref{lemmCrossingRelation}, between the $N$-th and $(N+1)$-th
crossing. Note that these two crossing sequences are not created by
real rotations $R_{0,\pm3}$ or $R_{\pm3,0}$, but rather only equal
to respectively the crossing sequences created by these two
rotations. Hence, according to the assumption that the claim is
valid for an $N$-crossing braid with arbitrary end-node states, and
the fact of the validity for all one-crossing braids, we arrive at
\begin{equation*}
X_l(S_l,\pm1)x_1x_2...x_Nx_{N+1}X_r(S_r,\mp1)=(\bar{x}_1\bar{x}_2...\bar{x}_N)(\bar{x}_{N+1})=\bar{x}_1\bar{x}_2...\bar{x}_{N+1}.
\end{equation*}

Bearing in mind that $|X|=|\bar{X}|$, therefore by induction,
simultaneous rotations, $R_{\pm 3,\mp 3}$ take a generic braid with
crossing sequence $X$ to an equivalent braid with sequence
$\bar{X}$, which does not change the number of crossings. This
certainly indicates that $R_{3k,-3k}$ with $k\in\mathbb{Z}$ rotates
an extremum to another extremum of the same braid, which validates
the proof. Furthermore, by Eq. \ref{eqRotAlgebraic1}, Eq.
\ref{eqRotAlgebraic2}, and Eq. \ref{eqPermrelation}, and with a
convenient redefinition: $\bar{X}=\bar{f}(X)$, we can pin down the
algebraic form of the action of a $R_{3k,-3k},\ k\in\mathbb{Z}$ on
generic braids,
\begin{eqnarray}
&&R_{3k,-3k}(_{T_l}^{S_l}\hspace{-0.5 mm}[ (T_a,T_b,T_c)\sigma_X] _{T_r}^{S_r})\nonumber\\
&&=_{\ \ T_l+ 3k}^{(-)^{3k}S_l}\hspace{-1.5 mm}[
((1,3)^k(T_a,T_b,T_c))\sigma_{\bar{f}^k(X)}]
_{T_r- 3k}^{(-)^{3k}S_r},%
\label{eqSimpirotAlgebraic1}
\end{eqnarray}
or equivalently,
\begin{eqnarray}
&&R_{3k,-3k}(_{T_l}^{S_l}\hspace{-0.5 mm}[\sigma^{-1}_X (T_{a'},T_{b'},T_{c'})] _{T_r}^{S_r})\nonumber\\
&&=_{\ \ T_l+ 3k}^{(-)^{3k}S_l}\hspace{-1.5 mm}[
\sigma^{-1}_{\bar{f}^k(X)}((1,3)^k(T_{a'},T_{b'},T_{c'}))] _{T_r-
3k}^{(-)^{3k}S_r},%
\label{eqSimpirotAlgebraic2}
\end{eqnarray}
where $(1,3)^k$ is the permutation induced by $R_{3k,-3k}$, and
$\bar{f}^k(X)=X$ for $k$ even, while $\bar{f}^k(X)=\bar{f}(X)$ for
$k$ odd.
\end{proof}

Similar to the case of actively-interacting braids, there are
conserved quantities under rotations in the form of $R_{3k,-3k}$,
which are shown in the Lemma below.

\begin{lemma}
$T_l+T_r$, $\sum\limits_{i=a}^{c}T_i$, and
$\sum\limits_{i=1}^{|X|}x_i$ of an extremum of a braid are invariant
under rotations of the form $R_{3k,-3k},\ k\in\mathbb{Z}$. The
triple $(T_a,T_b,T_c)$ is invariant when $k$ is even.
\label{lemmConsSimRot3k}
\end{lemma}
\begin{proof}
From Eq. \ref{eqSimpirotAlgebraic1} and Eq.
\ref{eqSimpirotAlgebraic2}, it is obvious that a rotation of the
form $R_{3k,-3k},\ k\in\mathbb{Z}$ takes $T_l+T_r$ to
$T_l+3k+T_r-3k=T_l+T_r$, turns the triple $(T_a,T_b,T_c)$ into
$(1,3)^k(T_a,T_b,T_c)$ which is again $(T_a,T_b,T_c)$ if $k$ is even
but does not affect $\sum\limits_{i=a}^{c}T_i$ in any case, and
changes $\sum\limits_{i=1}^{|X|}X_i$ to
$\sum\limits_{i=1}^{|\bar{X}|}\bar{x}_i=\sum\limits_{i=1}^{|X|}x_i$.
\end{proof}

These conserved quantities imply that the previously defined
$\Theta$ is also a conserved quantity under this class of
equivalence moves, which is expected because in Lemma
\ref{lemmConserve} we have shown that $\Theta$ is conserved under
any rotations. That is, it is the same for any representation of a
braid.

Now that we have found conserved quantities under rotations of the
form $R_{3k,-3k},\ k\in\mathbb{Z}$. The only issue left behind is
that we have not proven that the simultaneous rotations of multiple
of $\pi$ are the only possible class of rotations under which the
set of all extrema of a braid is closed. If this is true, then each
conserved quantity in Lemma \ref{lemmConsSimRot3k} is identical for
all extrema of a braid. There are strong evidences that this is
indeed the case; however, we are lack of a rigorous proof.
Therefore, we only state this observation as a conjecture.
\begin{conjecture}
Any rotation that transforms an extremum to another extremum of the
same braid must take the form of $R_{3k,-3k}$ with $k\in\mathbb{Z}$.
Assuming this, all extrema of a braid share the same $T_l+T_r$,
$\sum\limits_{i=a}^{c}T_i$, and $\sum\limits_{i=1}^{|X|}x_i$.
\label{conj3KsimRot}
\end{conjecture}

\section{Algebra of interactions: symmetries and relations}
\label{secInt}%
In \cite{HackettWan2008} it is shown that the interaction of any two
active braids produces another active braid. Now that we are dealing
with not only active braids but also the passive ones, one may ask
what the outcome of the interaction of an active braid and a passive
braid, say a propagating braid, should be. To answer this question,
we need sufficient preparation, divided into the following
subsections.
\subsection{Conserved quantities under interactions}
We first repeat in words the interaction condition formulated in
\cite{LeeWan2007, HackettWan2008}. This condition demands that one
of the two braids, say $B_1$ and $B_2$, under an interaction must be
active and that the two adjacent nodes, one of $B_1$ and the other
of $B_2$, are either already in or can be rotated to the
configuration where they have the same state and share a twist-free
edge. The latter requirement is actually the condition of a
$2\rightarrow 3$ Pachner move\cite{LeeWan2007}. The algebraic form
of this condition is explicitly given in \cite{HackettWan2008} and
is thus not duplicated here but will rather be adopted directly.

Since a braid is in fact an equivalence class of braid diagrams, a
convenient choice of the representative of the class is important.
Whether a braid is propagating or not is most transparent when the
braid is represented by its unique representative which has no
external twist. On the other hand, an active braid can always be put
in a trivial representation, which simplifies the calculation of
interactions. Therefore, in this section any active braid is
represented by one of its extrema, i.e. a trivial diagram with
external twists, and any braid which does not actively interact is
represented by its unique representative.

Now when an active braid, $B$, meets a passive braid, $B'$, say from
the left of $B'$ (the case where $B$ is on the right of $B'$ follows
similarly), with the interaction condition fulfilled, what does the
resulted braid $B+B'$ look like? Here, as in \cite{HackettWan2008}
we use a $+$ for the operation of interaction. A special case is
that $B$ in its trivial form has no right external twist, and since
$B'$ is in the representation without external twist, one can
directly apply a $2\rightarrow 3$ move of the $B$'s right end-node
and the left end-node of $B'$, then paly with the techniques
introduced in \cite{LeeWan2007} to complete the interaction. Let us
address this simple case first.
\begin{lemma}
Given an active braid
$B=\left._{T_l}^{\hspace{0.5mm}S}\hspace{-0.5mm}[T_a,T_b,T_c]^{S}_{T_r}\right.$,
with $T_r=0$, and a passive braid
$B'=\left._{\hspace{1mm}0}^{S_l}\hspace{-0.5mm}[(T'_a,T'_b,T'_c)\sigma_X]^{S_r}_{0}\right.$,
with $S_l=S$, the interaction of $B$ and $B'$ with $B$ on the left
of $B'$ produces
$B''=\left._{\hspace{8.5mm}0}^{(-)^{T_l}S_l}\hspace{-0.5mm}[(P^S_{-T_l}(T_a+T'_a+T_l,T_b+T'_b+T_l,T_c+T'_c+T_l))\sigma_{X_l(S,-T_l)X}]^{S_r}_{0}\right.$.
\label{lemmIntNotwist}
\end{lemma}
\begin{proof}
As $T_r=0$ and $S=S_l$, the interaction condition is met and thus no
rotation is needed; hence, according to \cite{LeeWan2007}, $B+B'$
forms a connected sum of $B$ and $B'$, which is, in our algebraic
language,
\begin{equation}
\begin{aligned}
B+B' &=
\left._{T_l}^{S_l}\hspace{-0.5mm}[T_a,T_b,T_c]^{S_l}_{0}\right. \#
\left._{\hspace{1mm}0}^{S_l}\hspace{-0.5mm}[(T'_a,T'_b,T'_c)\sigma_X]^{S_r}_{0}\right.\\
&=\left._{T_l}^{S_l}\hspace{-0.5mm}[(T_a+T'_a,T_b+T'_b,T_c+T'_c)\sigma_X]^{S_r}_{0}\right.\\
&\cong R_{-T_l,0}\left(
\left._{T_l}^{S_l}\hspace{-0.5mm}[(T_a+T'_a,T_b+T'_b,T_c+T'_c)\sigma_X]^{S_r}_{0}\right.
\right)\\
&=\left._{\hspace{8.5mm}0}^{(-)^{T_l}S_l}\hspace{-0.5mm}[(P^S_{-T_l}(T_a+T'_a+T_l,T_b+T'_b+T_l,T_c+T'_c+T_l))\sigma_{X_l(S,-T_l)X}]^{S_r}_{0}\right.,
\end{aligned}
\label{eqIntNoTwist}
\end{equation}
where a rotation $R_{-T_l,0}$ is applied after the connected sum to
put the resulted braid in its representative with zero external
twist, which induces a permutation $P^S_{-T_l}$ on the left triple
of internal twists, and a crossing sequence $X_l(S,-T_l)$, appended
to the original $X$ from left.
\end{proof}

However, in general the trivial diagram representing an active braid
may have external twists on both external edges. If the interaction
condition is satisfied when the trivial braid in this case meets a
passive braid, a rotation is usually required in order to perform
the connected sum algebraically for them to interact. We now deal
with this.
\begin{lemma}
Given an active braid
$B=\left._{T_l}^{\hspace{0.5mm}S}\hspace{-0.5mm}[T_a,T_b,T_c]^{S}_{T_r}\right.$
on the left of a passive braid,
$B'=\left._{\hspace{1mm}0}^{S_l}\hspace{-0.5mm}[(T'_a,T'_b,T'_c)\sigma_X]^{S_r}_{0}\right.$,
with the interaction condition satisfied by $(-)^{T_r}S=S_l$, the
interaction of $B$ and $B'$ results in a braid
$$B''=\left._{\hspace{7.5mm}0}^{(-)^{T_l}S}\hspace{-0.5mm}
[((P^S_{-T_l}(T_a,T_b,T_c))+(P^{(-)^{T_r}S}_{-T_l-T_r}(T'_a,T'_b,T'_c))+(T_l+T_r,\cdot,\cdot))\sigma_{X_l((-)^{T_r}S,-T_l-T_r)X}]^{S_r}_{0}\right.,$$
where $(T_l+T_r,\cdot,\cdot)$ is the short for
$(T_l+T_r,T_l+T_r,T_l+T_r)$.
\label{lemmIntGen}
\end{lemma}
\begin{proof}
\begin{align}
& B+B'\nonumber\\
&=
\left._{T_l}^{\hspace{0.5mm}S}\hspace{-0.5mm}[T_a,T_b,T_c]^{S}_{T_r}\right.
+
\left._{\hspace{1mm} 0}^{S_l}\hspace{-0.5mm}[(T'_a,T'_b,T'_c)\sigma_X]^{S_r}_{0}\right.\nonumber\\
&\cong
R_{T_r,-T_r}\left(\left._{T_l}^{\hspace{0.5mm}S}\hspace{-0.5mm}[T_a,T_b,T_c]^{S}_{T_r}\right.
\right) \#
\left._{\ 0}^{S_l}\hspace{-0.5mm}[(T'_a,T'_b,T'_c)\sigma_X]^{S_r}_{0}\right.\nonumber\\
&= \left._{\
T_l+T_r}^{(-)^{T_r}S}\hspace{-0.5mm}[P^S_{T_r}(T_a,T_b,T_c)]^{(-)^{T_r}S}_{0}\right.
\# \left._{\
0}^{S_l}\hspace{-0.5mm}[(T'_a,T'_b,T'_c)\sigma_X]^{S_r}_{0}\right.\label{eqAtrivB}\\
&= \left._{\
T_l+T_r}^{(-)^{T_r}S}\hspace{-0.5mm}[((P^S_{T_r}(T_a,T_b,T_c))+(T'_a,T'_b,T'_c))\sigma_X]^{S_r}_{0}\right.
\ \ \ \ \ \ \xLeftarrow{(-)^{T_r}S=S_l} \label{eqForTheoNoActive}\\
&\cong R_{-T_l-T_r,0}\left( \left._{\
T_l+T_r}^{(-)^{T_r}S}\hspace{-0.5mm}[((P^S_{T_r}(T_a,T_b,T_c))+(T'_a,T'_b,T'_c))\sigma_X]^{S_r}_{0}\right.
\right)\nonumber\\
&=
\left._{\hspace{7.5mm}0}^{(-)^{T_l}S}\hspace{-0.5mm}
[(P^{(-)^{T_r}S}_{-T_l-T_r}((P^S_{T_r}(T_a,T_b,T_c))+(T'_a,T'_b,T'_c))+(T_l+T_r,\cdot,\cdot))\sigma_{X_l((-)^{T_r}S,-T_l-T_r)X}]^{S_r}_{0}\right.\nonumber\\
&= \left._{\hspace{7.5mm}0}^{(-)^{T_l}S}\hspace{-0.5mm}
[((P^S_{-T_l}(T_a,T_b,T_c))+(P^{(-)^{T_r}S}_{-T_l-T_r}(T'_a,T'_b,T'_c))+(T_l+T_r,\cdot,\cdot))\sigma_{X_l((-)^{T_r}S,-T_l-T_r)X}]^{S_r}_{0}\right.\label{eqAgenerB},
\end{align}
where the simultaneous rotation $R_{T_r,-T_r}$, is applied to
realize the interaction condition in order to do the connected sum,
and the rotation $R_{-T_l-T_r,0}$ is exerted such that the final
result is in the representative without external twists, which
induces a permutation $P^{(-)^{T_r}S}_{-T_l-T_r}$ and a crossing
sequence $X_l((-)^{T_r}S,-T_l-T_r)$ concatenated to $X$ from left.
The above equation obviously reduces to Eq. \ref{eqIntNoTwist} when
$T_r=0$.
\end{proof}

Nevertheless, for each active braid, there are infinite number of
trivial braid diagrams which are equivalent, in the sense that any
two of them are related by a simultaneous rotation $R_{n,-n}$, $n\in
\mathbb{Z}$. It is then naturally to ask if the choice of the
trivial braid diagram representing a braid equivalence class
influences the interaction of the braid and another braid. The
answer is "No". The reason is obvious because of the equivalence of
the trivial diagrams. However, due to the necessity of realizing of
the interaction condition in a concrete calculation of interaction,
it is better to formulate this claim explicitly in our new notation
as a Lemma.
\begin{lemma}
For any active braid $B$, its interaction onto any other braid $B'$,
i.e. $B+B'$ or $B'+B$, is independent of the choice of the trivial
diagram representing $B$. \label{lemmIntIndepTrivB}
\end{lemma}
\begin{proof}
We prove the case of $B+B'$. Let
$B_0=\left._{T_l}^{\hspace{0.5mm}S}\hspace{-0.5mm}[T_a,T_b,T_c]^{S}_{T_r}\right.$
be a trivial diagram representing an active braid $B$. Let
$B'=\left._{\hspace{1mm}0}^{S_l}\hspace{-0.5mm}[(T'_a,T'_b,T'_c)\sigma_X]^{S_r}_{0}\right.$
be the passive braid on which $B$ interacts. We assume the
interaction condition is satisfied by $(-)^{T_r}S=S_l$. Any other
trivial braid, say $B_n$, representing $B$ can be obtained from
$B_0$ by
$$B_n=R_{n,-n}(B_0)=\left._{\hspace{1.4mm}T_l+n}^{(-)^nS}\hspace{-0.5mm}[P^S_n(T_a,T_b,T_c)]^{(-)^nS}_{T_r-n}\right..$$
$B_0+B'$ has already been shown in Lemma \ref{lemmIntGen}, but we
only need Eq. \ref{eqAtrivB} therein, which is the configuration of
the two braid after the interaction condition is realized. If we
replace $B_0$ by $B_n$ in the interaction, we have
\begin{align*}
B_n+B' &=
\left._{\hspace{1.4mm}T_l+n}^{(-)^nS}\hspace{-0.5mm}[P^S_n(T_a,T_b,T_c)]^{(-)^nS}_{T_r-n}\right.
+
\left._{\hspace{1mm}0}^{S_l}\hspace{-0.5mm}[(T'_a,T'_b,T'_c)\sigma_X]^{S_r}_{0}\right.\\
&\cong R_{T_r-n,-T_r+n}\left(
\left._{\hspace{1.4mm}T_l+n}^{(-)^nS}\hspace{-0.5mm}[P^S_n(T_a,T_b,T_c)]^{(-)^nS}_{T_r-n}\right.
\right) \#
\left._{\hspace{1mm}0}^{S_l}\hspace{-0.5mm}[(T'_a,T'_b,T'_c)\sigma_X]^{S_r}_{0}\right.\\
\xRightarrow{(-)^{2n}S=S} &= \left._{\
T_l+T_r}^{(-)^{T_r}S}\hspace{-0.5mm}[P^{(-)^nS}_{T_r-n}P^S_{n}(T_a,T_b,T_c)]^{(-)^{T_r}S}_{0}\right.
\# \left._{\
0}^{S_l}\hspace{-0.5mm}[(T'_a,T'_b,T'_c)\sigma_X]^{S_r}_{0}\right.\\
\xRightarrow{P^{(-)^nS}_{T_r-n}=P^{S}_{T_r}P^{(-)^nS}_{-n}} &=
\left._{\
T_l+T_r}^{(-)^{T_r}S}\hspace{-0.5mm}[P^{S}_{T_r}P^{(-)^nS}_{-n}P^S_{n}(T_a,T_b,T_c)]^{(-)^{T_r}S}_{0}\right.
\# \left._{\
0}^{S_l}\hspace{-0.5mm}[(T'_a,T'_b,T'_c)\sigma_X]^{S_r}_{0}\right.\\
\xRightarrow{P^{(-)^nS}_{-n}P^S_{n}\equiv\mathds{1}\ \text{by Lemma
\ref{lemmPermRel}}} &= \left._{\
T_l+T_r}^{(-)^{T_r}S}\hspace{-0.5mm}[P^S_{T_r}(T_a,T_b,T_c)]^{(-)^{T_r}S}_{0}\right.
\# \left._{\
0}^{S_l}\hspace{-0.5mm}[(T'_a,T'_b,T'_c)\sigma_X]^{S_r}_{0}\right.,
\end{align*}
which is exactly the same as Eq. \ref{eqAtrivB}. That is, if $B_0$
interacts with $B'$, so does $B_n$, and they give rise to the same
result. Likewise, this is also true for the case of $B'+B$. This
closes the proof.
\end{proof}

Now that we established Lemma \ref{lemmIntIndepTrivB}, we may choose
to always represent an active braid $B$ by its trivial
representative without right (left) external twist, in dealing with
the interaction of $B$ onto a passive braid from the left (right),
which simplifies the calculation and expression because Lemma
\ref{lemmIntNotwist} directly applies. Moreover, the result of Lemma
\ref{lemmIntNotwist}, namely Eq. \ref{eqIntNoTwist}, is identical to
the result when the active braid is represented by its unique
representative with zero external twists. This again, together with
\cite{HackettWan2008}, shows that the choice of representative of a
braid does not affect the result of the interaction involving the
braid, in accordance with \cite{LeeWan2007}. Examples can be found
in \cite{LeeWan2007}, one just need to cast them in our new symbolic
notation.

Equipped with this algebra, we shall prove one of our primary
results.
\begin{theorem}
Given an active braid,
$B=\left._{T_l}^{\hspace{0.5mm}S}\hspace{-0.5mm}[T_a,T_b,T_c]^{S}_{T_r}\right.$,
and a passive braid,
$B'=\left._{\hspace{1mm}0}^{S_l}\hspace{-0.5mm}[(T'_a,T'_b,T'_c)\sigma_X]^{S_r}_{0}\right.$,
such that $B''=B+B'$, the effective twist number $\Theta$ is an
additive conserved quantity, while the effective state $\chi$ is a
multiplicative conserved quantity, namely%
\begin{equation}
\begin{aligned}
\Theta_{B''} &=\Theta_B+\Theta_{B'}\\
\chi_{B''} &=\chi_B\chi_{B'}.
\end{aligned}
\end{equation}
\label{theoConserve}
\end{theorem}
\begin{proof}
We can readily write down
$$\Theta_B=T_l+T_r+\sum\limits_{i=a}^{c}T_i,\ \ \chi_B=1,$$ and
$$\Theta_{B'}=\sum\limits_{i=a}^{c}T_i-2\sum\limits_{j=1}^{|X|}x_j,\ \ \chi_{B'}=(-)^{|X|}S_lS_r.$$
Hence, according to Eq. \ref{eqAgenerB}, we have
\begin{align}
\Theta_{B''} &=\sum\limits_{i=a}^{c}(T_i+T'_i)+3(T_l+T_r)-2\sum\limits_{j=1}^{|X_l|}x_j-2\sum\limits_{k=1}^{|X|}x_k\nonumber\\
&=
\sum\limits_{i=a}^{c}(T_i+T'_i)+(T_l+T_r)-2\sum\limits_{k=1}^{|X|}x_k\nonumber\\
&=\Theta_B+\Theta_{B'}, \label{eqSumTheta}
\end{align}
where the second equality is a result of
$\sum\limits_{j=1}^{|X_l((-)^{T_r}S,-T_l-T_r)|}x_j=T_l+T_r$, by Eq.
\ref{eqRotXvalue}. Besides,%
\begin{equation}
\chi_{B''}=(-)^{|T_l+T_r|+|X|}(-)^{T_l}SS_r=(-)^{|X|}(-)^{T_r}SS_r=(-)^{|X|}S_lS_r=\chi_{B'}=\chi_B\chi_{B'}.%
\eeq
\end{proof}

This theorem demonstrates that the by far discovered two
representative-independent conserved quantities, $\Theta$ and $\chi$
are also conserved under interactions, in the sense that the former
is additive while the latter is multiplicative. This is consistent
to \cite{HackettWan2008} in which only interactions between active
braids are discussed. In particular, $\chi$ becomes the $S^2$ in
\cite{HackettWan2008}, whose conservation means that the interacting
character of the braids is preserved. Furthermore, according to
Theorem \ref{theoChi}, the multiplicative conservation of $\chi$
shows that if the passive braid involved in an interaction has
$\chi=-1$, the resulted braid must also has $\chi=-1$, and is thus a
passive braid too.

\subsection{Asymmetry between $B+B'$ and $B'+B$}
\label{subsecAsymm}%
It is important to note that Theorem
\ref{theoConserve} is also true for the case of interaction where
the active braid is on the right of the passive braid. In fact, all
the discussion above can be equally well applied to this case. One
must then ask a question: does an active braid gives the same result
when it interacts on to a passive braid from the left and from the
right respectively? The answer is "No" in general. We now discuss
this issue by considering an active braid, $B$, and a passive braid,
$B'$.

First of all, even if the interaction condition is met in the case
$B+B'$, there is no guarantee that the interaction condition is also
satisfied in the case of $B'+B$, which means $B'+B$ is simply an
impossible interaction. If we assume the interaction condition can
be realized in both cases, $B+B'$ and $B'+B$ still give rise to
different results, i.e. inequivalent braids, in general. Let us
study the detail.

In the case of $B+B'$, by Lemma \ref{lemmIntIndepTrivB} we can
represent $B$ by its trivial braid diagram without $T_r$, viz
$B\cong B_l=\left._{T_l}^{\hspace{0.5mm}
S}\hspace{-0.5mm}[T_a,T_b,T_c]^{S}_{0}\right.$, which allows us to
use Eq. \ref{eqIntNoTwist}. However, in the case of $B'+B$, we
represent $B$ by its trivial diagram without the left external
twist, which is obtained by a simultaneous rotation from $B_l$, i.e.
$B\cong B_r=R_{-T_l,T_l}(B_l)=\left._{\hspace{7.5mm}0}^{
(-)^{T_l}S}\hspace{-0.5mm}[P^S_{-T_l}(T_a,T_b,T_c)]^{(-)^{T_l}S}_{T_l}\right.$.
Then for
$B'=\left._{\hspace{1mm}0}^{S_l}\hspace{-0.5mm}[(T'_a,T'_b,T'_c)\sigma_X]^{S_r}_{0}\right.$,
the interaction condition in the latter case is $S_r=(-)^{T_l}S$.
With also $S=S_l$, we conclude that the condition for both $B+B'$
and $B'+B$ doable is%
\begin{equation}
S_l=(-)^{T_l}S_r. %
\label{eqDoubleIntCond}
\end{equation}
Given this, by a similar calculation as that in Lemma
\ref{lemmIntNotwist}, one
can find %
\begin{align}
B'+B&
=\left._{\hspace{1mm}0}^{S_l}[((T'_a+T_l,T'_b+T_l,T'_c+T_l)+(\sigma^{-1}_X
P^S_{-T_l}(T_a,T_b,T_c)))\sigma_{XX_r(S_r,-T_l)}]^{(-)^{T_l}S_r}_0\right.\nonumber\\
&=\left._{\hspace{1mm}0}^{S_l}[((T'_a+T_l,T'_b+T_l,T'_c+T_l)+(\sigma^{-1}_X
P^S_{-T_l}(T_a,T_b,T_c)))\sigma_{XX_r(S_r,-T_l)}]^{S_l}_0\right..
\label{eqB'B}
\end{align}
To compare this to $B+B'$, we rewrite Eq. \ref{eqIntNoTwist} as
follows, taking Eq. \ref{eqDoubleIntCond} into account.%
\begin{equation}
B+B'=\left._{\hspace{1.5mm}0}^{S_r}\hspace{-0.5mm}[(P^S_{-T_l}(T_a+T'_a+T_l,T_b+T'_b+T_l,T_c+T'_c+T_l))
\sigma_{X_l(S_l,-T_l)X}]^{S_r}_{0}\right. \label{eqBB'}
\end{equation}
It is important to notice that, by Eq. \ref{eqB'B} and Eq.
\ref{eqBB'}, both $B+B'$ and $B'+B$ are in the representative with
zero external twists. Since we know such a representative is unique
for each equivalence class of braids, there is no way for $B+B'$ and
$B'+B$ to be equivalence to each other. But there are possibilities
for $B+B'$ and $B'+B$ are simply equal. For this to be true, braids
$B$ and $B'$ are pretty strongly constrained.

Firstly, one obviously has to require $S_l=S_r$ and $T_l=2k,\
k\in\mathbb{Z}$ in Eq. \ref{eqB'B} and Eq. \ref{eqBB'}, such that
$B+B'$ and $B'+B$ have the same end-node states. This, together with
the interaction condition, also indicates $S_l=S_r=S$. Keeping this
in mind, one must then demand that $B+B'$ and $B'+B$ have the same
crossing sequence, namely $X_l(S,-T_l)X=XX_r(S,-T_l)$, which can
also be written as $X_l(S,-T_l)XX^{-1}_r(S,-T_l)=X$. With Lemma
\ref{lemmCrossingRelation}, this can be put in the form,
$X_l(S,-T_l)XX_l(S,T_l)=X$. Since Eq. \ref{eqCrossing} reads that
$X_l(S,m)=X_r(S,m)$ for $m$ even, this condition is better expressed
as $$X_l(S,-T_l)XX_r(S,T_l)=X,$$ which now appears to be the
requirement that a simultaneous $(-T_l)$-rotation leaves the
crossing sequence intact. Although it has been conjectured in the
end of Section \ref{secEquiv} that only simultaneous rotations of
$6k,\ k\in\mathbb{Z}$ are able to achieve this condition, we would
like to keep its current general format because the conjectured has
not been proved. However, interestingly, we will see an automatic
input of $6k,\ k\in\mathbb{Z}$ shortly.

From this point on, three possibilities will arise, each of which
leads to a different condition on top of the conditions found above.
The first possibility is that we do not put any constraint on the
internal twists. Hence, according to Eq. \ref{eqB'B} and Eq.
\ref{eqBB'}, such that $B'+B$ equals $B+B'$ we have to demand
$\sigma^{-1}_XP^S_{-T_l}=P^S_{-T_l}\equiv\mathds{1}$. An immediate
consequence is that $\sigma_X=\mathds{1}$, which however does not
constraint the pattern of $X$ very much. Moreover, as we know that
$T_l$ is even, then by Eq. \ref{eqPermrelation},
$P^S_{-T_l}\equiv\mathds{1}$ if and only if $T_l=6j,\
j\in\mathbb{Z}$. With this $T_l$, Lemma \ref{lemm3KsimRot} ensures
the fulfillment of the condition, $X_l(S,-T_l)XX_r(S,T_l)=X$, which
now appears to be redundant in this case.

The second possibility is to relax the first one a little bit by
only requiring $\sigma^{-1}_XP^S_{-T_l}=P^S_{-T_l}$ which only means
that $\sigma_X=\mathds{1}$ but no criteria for $P^S_{-T_l}$. Then we
notice that while in Eq. \ref{eqBB'} both of the triples
$(T_a,T_b,T_c)$ and $(T'_a,T'_b,T'_c)$ are under the same
permutation $P^S_{-T_l}$, in Eq. \ref{eqB'B} only the triple
$(T_a,T_b,T_c)$ is permuted by $P^S_{-T_l}$. Therefore, the only way
to make $B+B'=B'+B$ is to mandate
$P^S_{-T_l}(T'_a,T'_b,T'_c)=(T'_a,T'_b,T'_c)$. Note that this does
not necessarily mean $P^S_{-T_l}=\mathds{1}$, e.g. when
$P^S_{-T_l}=(1,\ 2)$ and $T'_a=T'_b$.

The last possibility is to remove even the condition
$\sigma_X=\mathds{1}$. As a result, the constraint on internal
twists turns out to be stronger than those in the previous two
possibilities. It is not hard to see, we must have in general
$T_a=T_b=T_c$ but still
$P^S_{-T_l}(T'_a,T'_b,T'_c)=(T'_a,T'_b,T'_c)$ in the meanwhile.

We have now exhausted all conditions with the most general
consideration. More importantly, although we restricted our
discussion to the case where $B'$ is a passive braid, the conditions
automatically applies to the case where $B'$ is even an active braid
in its unique representation. Let us summarize this in the following
theorem as another primary result of this paper.
\begin{theorem}
Given an active braid
$B=\left._{T_l}^{\hspace{0.5mm}S}\hspace{-0.5mm}[T_a,T_b,T_c]^{S}_{0}\right.$,
and an arbitrary braid (passive or active) in its unique
representative, namely
$B'=\left._{\hspace{1mm}0}^{S_l}\hspace{-0.5mm}[(T'_a,T'_b,T'_c)\sigma_X]^{S_r}_{0}\right.$,
active or passive, for $B+B'=B'+B$ to be true, we demand%
\begin{equation}
\begin{aligned}
&S_l=S_r=S\\
&T_l=2k,\ k\in\mathbb{Z}%
\end{aligned}
\end{equation}
and%
\begin{equation}
X_l(S,-T_l)XX_r(S,T_l)=X%
\label{eqXingCond}
\end{equation}
and any of the following three:
\begin{enumerate}
\item
    \begin{align}
    &\sigma_X=\mathds{1}\\
    &T_l=6k,\ k\in\mathbb{Z}\label{eqTl6k}
    \end{align}%
\item
    \begin{equation}
    \begin{aligned}
    &\sigma_X=\mathds{1}\\
    &P^S_{-T_l}(T'_a,T'_b,T'_c)=(T'_a,T'_b,T'_c)
    \end{aligned}
    \end{equation}
\item
    \begin{equation}
    \begin{aligned}
    &T_a=T_b=T_c\\
    &P^S_{-T_l}(T'_a,T'_b,T'_c)=(T'_a,T'_b,T'_c).
    \end{aligned}
    \end{equation}
\end{enumerate}%
\label{theoBB'=B'B}
\end{theorem}
An important remark is that Theorem \ref{theoBB'=B'B} is based on
the assumption that Conjecture \ref{conj3KsimRot} may be incorrect.
If Conjecture \ref{conj3KsimRot} happen to be correct (there is a
strong evidence that it is indeed so), then the conditions in
Theorem \ref{theoBB'=B'B} should be modified as follows. The
satisfaction of Eq. \ref{eqXingCond} indicates $T_l=6k,\
k\in\mathbb{Z}$, which immediately ensures that Eq. \ref{eqTl6k}
holds and that $P^S_{-T_l}\equiv\mathds{1}$. Therefore, by simple
logic the condition that $B+B'=B'+B$, if Conjecture
\ref{conj3KsimRot} stands, is reduced to:
\begin{equation}
\begin{aligned}
&S_l=S_r=S\\
&T_l=6k,\ k\in\mathbb{Z}%
\end{aligned}
\end{equation}
and either
\begin{equation}
\sigma_X=\mathds{1}
\end{equation}
or
\begin{equation}
T_a=T_b=T_c.
\end{equation}

The discussion above focuses on how $B+B'$ is equal to $B'+B$.
Nevertheless, since \cite{HeWan2008a} discovered discrete
transformations of braids, which are mapped to $\mathcal{C}$,
$\mathcal{P}$, $\mathcal{T}$, and their products, and discussed
their actions on braid interactions, $B+B'$ and $B'+B$ may be just
related by a discrete transformation. As pointed out in
\cite{HeWan2008a}, the transformations $\mathcal{P}$, $\mathcal{T}$,
$\mathcal{CP}$, and $\mathcal{CT}$ swap the two braids undergoing an
interaction. That is, for example,
$\mathcal{P}(B+B')=\mathcal{P}(B')+\mathcal{P}(B)$. If $B$ and $B'$
happen to be invariant under $\mathcal{P}$, $B'+B$ is then equal to
$\mathcal{P}(B+B')$. We will not involve the detailed mathematical
formats of these transformations, which can be found in
\cite{HeWan2008a}; rather, we list below the conditions for this to
be true.
\begin{equation}
B'+B=\left\{ %
\begin{array}
[c]{ll}%
\mathcal{P}(B+B'),&\ \ B=\mathcal{P}(B),\ B'=\mathcal{P}(B') \medskip\\
\mathcal{T}(B+B'),&\ \ B=\mathcal{T}(B),\ B'=\mathcal{T}(B') \medskip\\
\mathcal{CP}(B+B'),&\ \ B=\mathcal{CP}(B),\ B'=\mathcal{CP}(B') \medskip\\
\mathcal{CT}(B+B'),&\ \ B=\mathcal{CT}(B),\ B'=\mathcal{CT}(B').
\end{array}
\right.
\end{equation}

\subsection{The algebraic structure}
With the help of this notation and the algebraic method established
on it we are able to show that the set of all stable braids, namely
the active braids, propagating braids, and stationary braids, form
an algebra under the braid interaction. This algebra is closed. The
reason is that any interaction of the type defined in
\cite{LeeWan2007} of two stable braids never leads to an instable
braid due to the stability condition put forward in
\cite{LeeWan2007, Isabeau2008}.

In \cite{HackettWan2008} it is demonstrated that an interaction
between two active braids always results in another active braid. On
the other hand, interactions between active and passive braids turn
out to be more complicated and involved. However, provided with all
the discussion in previous sections, we can try to answer the
question raised at the beginning of Section \ref{secInt}. This
question can be first partly answered by the following Theorem.

\begin{theorem}
The resulted braid of any interaction between an active braid and a
passive braid is again passive.%
\label{theoNoActive}
\end{theorem}
\begin{proof}
Recalling Theorem \ref{theoChi} and Theorem \ref{theoConserve},
since the effective state $\chi$ is a multiplicative conserved
quantity under interaction and an active braid must have $\chi=1$,
the interaction of an active braid and any passive braid with
$\chi=-1$ must leads to a passive braid with $\chi=-1$.

However, this is not a complete proof because a passive braid may
also have $\chi=1$. A full proof can be easily constructed. by
contradiction. For this purpose, we need the following facts,
extracted from \cite{LeeWan2007}, of passive braids:
\begin{enumerate}
\item A stationary braid is neither left nor right completely
reducible.
\item A (left-) right-propagating braid is never
completely-reducible from its (right) left end-node.
\item A two-way propagating braid is not completely-reducible from
either end-node; otherwise it must be both left and right
completely-reducible, which makes it an active braid if equipped
with appropriate twists.
\end{enumerate}

Now let us consider an active braid $B$ in an arbitrary trivial
representative and a passive braid $B'$ in its unique
representative, with the interaction condition met, their
interaction, say $B+B'$ (the case of $B'+B$, if possible, will
follow similarly), results in, by Eq. \ref{eqForTheoNoActive}, with
internal twists ignored because they are irrelevant,%
\begin{equation}
B+B'=\left._{\
T_l+T_r}^{\hspace{7mm}S_l}\hspace{-0.5mm}[(\cdot,\cdot,\cdot)\sigma_X]^{S_r}_{0}\right.%
\label{eqTheoNoActive}
\end{equation}
We do not apply a rotation of the left end-node of $B+B'$ because
things will become less transparent otherwise. Note that according
to Eq. \ref{eqTheoNoActive}, at this stage, the two end-nodes of
$B+B'$ are in the same states respectively as those of $B'$ before
the interaction, and also that it has the same crossing sequence $X$
as $B'$ does. This means the irreducibility of $B+B'$ respects that
of $B'$.

We know that an active braid must be both left and right
completely-reducible. Now that $B'$ is passive, it is never
completely-reducible from both sides, which means $B+B'$ is not
either by Eq. \ref{eqTheoNoActive}. Otherwise, $B'$ should be
two-way completely reducible in the first place, which is
contradictory to any basic facts listed above of a passive braid.
Therefore, the theorem holds.
\end{proof}

Since an interaction of two active braids gives rise to active
braids only, as aforementioned, Theorem \ref{theoNoActive} then
immediately entitles the set of active braids a subalgebra of the
algebra of stable braids.

We still need to discuss if the interaction of an active braid and a
passive propagating (stationary) braid creates a propagating
(stationary) braid. The answer is "Not always". To illustrate this,
we show two examples.

Let us consider an active braid,
$B=\left._{-1}^{\hspace{1.5mm}+}\hspace{-0.5mm}[1,1,-1]^{+}_{0}\right.$,
and a left-propagating braid,
$B_p=\left._{\hspace{0.7mm}0}^{+}\hspace{-0.5mm}[(-5,-5,1)\sigma_{ud^{-1}}]^{+}_{0}\right.$,
whose graphical presentations can be found in \cite{LeeWan2007,
HeWan2008a}. Hence, by Eq. \ref{eqIntNoTwist} we have%
\begin{align}
B+B_p&=\left._{\hspace{0.5mm}0}^{-}\hspace{-0.5mm}[(P^+_1(-5,-5,-1))\sigma_{u^{-1}ud^{-1}}]^{+}_{0}\right.\nonumber\\
\xRightarrow{P^+_1=(1,\ 2)}
&=\left._{\hspace{0.5mm}0}^{-}\hspace{-0.5mm}[(-5,-5,-1)\sigma_{d^{-1}}]^{+}_{0}\right.,\label{eqEx1}
\end{align}
which is an irreducible braid according to \cite{Wan2007}, and is
thus stationary. This example shows the interaction of an active
braid and a propagating braid can result in a stationary braid. The
reason for such a situation to arise is due to the pair cancelation
of crossings and the change of the end-node state, as a consequence
of the interaction, which is lucid in the example above.

On the other hand, An active braid and a stationary can also produce
a propagating braid via their interaction. For this sake, we can use
the braid in Eq. \ref{eqEx1} as the stationary one and name it
$B_s$, and consider an active braid
$B=\left._{\hspace{1mm}1}^{-}\hspace{-0.5mm}[-1,-1,1]^{-}_{0}\right.$.
Also by Eq. \ref{eqIntNoTwist} we obtain
\begin{align}
B+B_s &=\left._{\hspace{0.5mm}0}^{+}\hspace{-0.5mm}[(P^-_{-1}(-5,-5,1))\sigma_{ud^{-1}}]^{+}_{0}\right.\nonumber\\
\xRightarrow{P^-_{-1}=(1,\ 2)}
&=\left._{\hspace{0.5mm}0}^{+}\hspace{-0.5mm}[(-5,-5,1)\sigma_{ud^{-1}}]^{+}_{0}\right.,
\end{align}
which is the very $B_p$ in the previous example.

Above all, the set of all stable braids form an algebra with
interaction as the binary operation, which is associative, of the
algebra. Stable braids are local excitations of embedded spin
networks which are considered to be basis states describing the
fundamental space-time. Consequently, a physical state is usually a
superposition of these basis states. It is therefore clear that
braid interaction, as the binary operation of the Algebra of stable
braids, is bilinear. Within this algebra, the set of active braids
behaves as a subalgebra. In addition, due to the asymmetry between
left and right interactions elaborated in Section \ref{subsecAsymm},
this algebra is noncommutative.

\section{Conclusion, discussion and outlook}
Conservation laws play a pivotal role in revealing the underlying
structure of a physical theory. By means of invariants and conserved
quantities we are able to determine how the content of the theory
relates to particle physics, or what kind of new mathematical and/or
physical inputs are necessary so that the theory is meaningful.

We have generalized the algebraic notation of active braids,
proposed in \cite{HackettWan2008}, to all our braids, found a set of
equivalence relations relating them, and developed conserved
quantities associated with these relations. More importantly, by
means of this notation we studied the interaction between active
braids and passive braids, which leads us to the fact that the set
of all stable braids forms a noncommutative algebra, with a
subalgebra containing only active braids. From this we have found
both additive and multiplicative conserved quantities of braids
under interaction. These are not only dynamically conserved but also
conserved under the equivalence moves.

A possible next step is to determine which of these conserved
quantities may correspond to quantum numbers, together with the
results for interactions of braids found in this paper, to fully
classify the set of braids. These conserved quantities find an
application in a companion paper\cite{HeWan2008a} of us, in which
discrete transformations of our braids have been discovered and
mapped to charge conjugation, parity, time reversal and their
products. The results of interactions also stimulate another work in
this direction\cite{LeeWan2008}.

The ultimate physical content of our braids cannot be fully
understood at this stage. In \cite{LouNumber} regarding the braids
of 3-valent embedded spin networks, a tentative mapping between the
3-valent braids and Standard Model particles is proposed, with
however the absence of dynamics. In the 4-valent case, as also
discussed in the companion paper\cite{HeWan2008a}, such a direct
mapping, if not impossible, is at least obscure at this level of
understanding of the braids. A reason is that the dynamics, namely
the propagation and interaction of 4-valent braids, strongly
constraints the possible set of twists, crossing sequence, and
ene-node states of a braid for it to be propagating and/or
interacting. In addition, the closed form of this constraint is
still missing. Consequently, one is not supposed to assign a
4-valent braid any topological property just for it to be possibly a
Standard Model particle, which is what has been done in the 3-valent
case. More study and in particular maybe new mathematical tools are
needed to reveal whether the 4-valent can directly correspond to
Standard Model particles.

If our 4-valent braids are more fundamental entities than the
Standard Model particles are, then what do they correspond to, what
do their interactions mean actually, and how do they give rise to
Standard Model particles under certain semi-classical limit? These
are profound but natural questions to ask. However, our current
understanding of 4-valent braids has not provided sufficient
knowledge to give an answer. The realm of braids of spin networks is
enormous, and a great deal of future work must be done.

For example, in our study we have yet not included spin network
labels which are normally representations of gauge groups, or of the
quantum groups of the corresponding gauge groups. This may make one
misunderstand that the properties of braids are independent of the
spin network they live on. This is nevertheless not true. On the one
hand, although that a braid is propagating and/or interacting
depends on its topological setting, whether it can indeed propagate
away from its location and/or interact with its adjacent braids
depends on the structure of its neighborhood and hence of the whole
spin network it is on. On the other hand, when spin network labels
are taken into account, a braid becomes manifestly dependent of its
spin network, with only its topological properties unchanged. Braids
of the same topology but different set of spin network labels would
be considered physically different, though maybe not different
particles.

Moreover, with spin network labels, a dynamical move, e.g. a
$2\rightarrow 3$ move, may have a superposition of outcomes in
identical topological configuration but different set of spin
network labels; each outcome has a certain probability amplitude.
However, the original set of topological quantities of a braid which
is essential for the braid to be propagating and interacting is
still valid even after spin network labels are considered.

We know that the very notion of particles in local quantum field
theories depends on the background geometry of the theories. Our
braids also depend on the spin network they live on. In spite of the
fact that how to take a physically meaningful semi-classical limit
of our approach remains an open question, the matter particles
resulting from the braids in a semi-classical background as a
reasonable such limit of superposed spin networks would be expected
to depend on the background geometry as well.

It is therefore very interesting and necessary to study the effects
of spin networks in our future research. Our companion
paper\cite{HeWan2008a} also indicates, from another perspective, the
necessity of taking into account spin network labels.

\section*{Acknowledgements}
We are grateful to Yong-Shi Wu for his helpful discussion,
especially his stimulation of a very interesting future direction
along this research line. YW is indebted to his Advisor, Lee Smolin,
for his encouragement and critical comments. Thanks also go to
Jonathan Hackett for discussions. SH is also grateful to the
hospitality of the Perimeter Institute. Research at Perimeter
Institute for Theoretical Physics is supported in part by the
Government of Canada through NSERC and by the Province of Ontario
through MRI. Research at Peking University is supported by NSFC
(nos.10235040 and 10421003).

\end{document}